\documentclass{eptcs}

\usepackage[utf8]{inputenc}

\usepackage{amsfonts}
\usepackage{amsmath, amssymb}
\usepackage{graphicx}
\usepackage{wrapfig}

\newcounter{mycount}[section]

\newtheorem{lemma}{Lemma}
\newtheorem{thm}{Theorem}

\newtheorem{proposition}{Proposition}
\newtheorem{corollary}{Corollary}

\newenvironment{proof}
{\vspace{0.3 em} \noindent {\bf Proof.}}{\QED {\vspace{0.7 em}}}

\def\squareforqed{$\Box$}
\def\QED{\ifmmode\squareforqed\else{\unskip\nobreak\hfil
\penalty50\hskip1em\null\nobreak\hfil\squareforqed
\parfillskip = 0pt\finalhyphendemerits = 0\endgraf}\fi}

\newcommand{\blank}{\mbox{$\flat$}}

\newcommand{\exptime}{\mbox{\textit{EXP}}}
\newcommand{\doublyExptime}{\mbox{\textit{2EXP}}}
\newcommand{\up}{\mbox{\textit{UP}}}
\newcommand{\coup}{\mbox{\textit{co-UP}}}

\newcommand{\trees}[1] {\mathcal{T}_{#1}}
\newcommand{\pZ} {\mbox{\textit{Player 0}}}
\newcommand{\pO} {\mbox{\textit{Player 1}}}

\title{Tree games with regular objectives}

\author{Marcin Przybyłko
\footnote{The author is supported by the \emph{Expressiveness of Modal Fixpoint Logics} project realized within the 5/2012 Homing Plus programme of the Foundation for Polish Science, co-financed by the European Union from the Regional Development Fund within the Operational Programme Innovative Economy (``Grants for Innovation'').}
\institute
 {
  $\begin{array}{c c}
    \text{University of New Caledonia} & \text{University of Warsaw} \\
    \text{Noumea, New Caledonia} & \text{Warsaw, Poland}\\
  \end{array}$
 }
\email{M.Przybylko@mimuw.edu.pl}
}

\def\titlerunning{Regular Branching Games}
\def\authorrunning{M. Przybyłko}

\begin{document}

\date{}

\maketitle
\paragraph{Abstract}
\begin{abstract}
We study \emph{tree games} developed recently by Matteo Mio as a game interpretation of~the probabilistic $\mu$-calculus.
With expressive power comes complexity. Mio showed that \emph{tree games} are able to encode Blackwell games and, consequently, are not determined under deterministic strategies.

We show that \emph{non-stochastic tree games} with objectives recognisable by so-called \emph{game automata} are determined under deterministic, finite memory strategies.
Moreover,  we give an elementary algorithmic procedure which,
for an arbitrary regular language L and a finite non-stochastic
tree game with a winning objective L decides if the game is
determined under deterministic strategies.

\end{abstract}


\section{Introduction}
\label{s:intro}

Tree  games were developed by Matteo Mio as a~framework in which one could provide a precise game semantics of a certain extension of probabilistic $\mu$-calculus (the logic $pL\mu^{\odot}$), cf. \cite{mio_thesis}.
This goal was achieved in the form of stochastic two-player meta-parity games, a special restriction of tree games that coincides with the logic in question.

Tree games generalise standard turn based games (see, e.g. \cite{ApGr11}, \cite{Chatterjee07_StochasticOmegaRegularGames}) by allowing the execution of a play to be split into concurrent, independently executed sub-games.
This is formalised by introducing a new kind of vertices called \emph{branching vertices}.
When a play reaches one of those vertices, the game automatically splits into several sub-games, one for each successor of~the currently visited branching vertex, and continues their execution independently.
With arbitrary objectives this feature may deprive players of some information and, therefore, result in the loss of determinacy under deterministic strategies, cf. \cite{mio_thesis}, section~4.1.

When studying two player games per se, we are mainly interested in two goals.
To find winning strategies, i.e., descriptions of moves of a player that will guarantee certain objectives.
Or, when that is not possible, to find relatively simple arguments implying that there are no such strategies.
In this paper we try to fulfil those goals in the setting of~the tree games with regular objectives, i.e., objectives that are expressed by non-deterministic tree automata,
and deterministic strategies.

We consider regular objectives for several reasons.
They allow us to describe non-trivial relations between the concurrent sub-games,
e.g., using regular tree languages we can request that exactly one sub-game achieves the objective.
They are powerful enough to deprive players of perfect information:
the reduction of Blackwell games presented in \cite{mio_thesis} is, in~fact, obtained by enriching the original payoff function with a condition that can be defined by a regular tree language.
Last but not least, they are defined by arguably simple, well behaved, and fairly well understood model of computation namely parity tree automata.

\paragraph{Our contribution.}
In this paper, we bring two results.
The first one  proves that regular objectives defined by so called \emph{game automata} (cf. \cite{DBLP:conf/fsttcs/DuparcFM11}, \cite{DBLP:conf/lics/FacchiniMS13}) retain the determinacy under deterministic strategies.
The second one provides a simple algorithm that, for arbitrary regular objectives, decides in doubly exponential time whether and which player has a winning strategy.
\vspace{-9pt}
\paragraph{Outline.}
Section~\ref{s:def} provides basic definitions.
In Section~\ref{s:main_thm} we state the main result and
describe the intuition behind the proof, which can be found in Section~\ref{s:main_lemma}.
In Section~\ref{s:reductions} we discuss expressive power of tree games, their relation to regular tree languages, and formulate a~simple property of a tree languages that causes the lack of determinacy under deterministic strategies.
Section~\ref{s:winning_strats} provides a simple automata-based algorithm to decide in
doubly exponential time whether either of~the players has a winning strategy.
Finally, in Section~\ref{s:end} we mention several possible directions of future research.
\vspace{-9pt}
\section{Basic definitions}
\label{s:def}

By $\gamma$ we understand the following polymorphic function
\vspace{-1pt}
$$
  \gamma[x,y,z] := \left\{
  \begin{array}{l l}
  y & \quad \text{ if } x \text{ is true}, \\
  z & \quad \text{ otherwise. }
  \end{array} \right.
  \vspace{-1pt}
$$

$\mathbb{N}$ is the set of natural numbers.
An alphabet $\Gamma \subseteq \mathbb{N}$ is any finite subset of natural numbers.
For technical purposes, we will always assume that $\Gamma$ contains special number $\blank$.
By $\Gamma^{*}$ (resp. $\Gamma^{\omega}$) we denote set of all finite (resp. infinite) sequences of elements from $\Gamma$.
$\Gamma^{+}$ is the set of all non-empty finite sequences.
For any two sequences $u,v$ we write $u \sqsubseteq v$ if u is a prefix of v.
We say that sequence $u \in \mathbb{N}^\omega$ is \emph{winning} if $\liminf\limits_{n \to \infty} u(n)$ is even.
Any sequence that is not winning is called \emph{loosing}.

We use \emph{record notation} throughout this article.
If $t = \langle t_1, t_2, \dots, t_k\rangle$ is a tuple, then by $t.t_n$ we denote component $t_n$.
By $t[y:= x]$ we denote tuple obtained from $t$ by changing component $y$ into $x$.

Whenever $E$ is a binary relation, $t$ is called a successor of $s$ if $\langle s,t \rangle \in E$. By $E(s)$ we denote the set of all successors of $s$.
For every set $L$, $\chi_{L}$ denotes the indicator function of set $L$, i.e., function 
$\chi_L(x) :=\gamma[x \in L, 1,0].$
\vspace{-9pt}
\paragraph{Labelled tree.}
A labelled tree $t$ is any function $t: 2^{*} \to \Gamma$ such that for every two words $u,w \in 2^{*}$ if $t(u) = \blank$, then $t(uw) = \blank$.
Elements of~the set $2^{*}:= \{ 0,1 \}^{*}$ are called nodes, and the set of all trees labelled with alphabet $\Gamma$, i.e. the set of functions $t: 2^{*} \to \Gamma$, is denoted $\trees{\Gamma}$.

Let $u \in 2^{*}$ be a node and $t$ be a tree, by $t.u$ we  denote the sub-tree of $t$ rooted in node $u$, i.e,
the tree $t_2$ such that $t_2(v) = t(uv)$.
If $t_2$ is a tree, then by $t[u:= t_2]$ we denote tree $t$ with sub-tree $t.u$ replaced by~$t_2$, i.e., $t[u:= t_2]$ satisfies
\[
t[u:= t_2](x) = \left\{ 
  \begin{array}{l l}
  t_2(v)          & \quad \text{if } x = uv,\\
  t(x)            & \quad \text{otherwise}.\\
  \end{array} \right.
\]
A \emph{restriction} of a tree $t$ is any tree $t'$ such that for every node $u \in 2^{*}$, $t'(u) = \blank$ or $t'(u) = t(u)$.
Node labelled $\blank$ is called a \emph{blank} node, and any sub-tree with only blank nodes is called a \emph{blank} sub-tree.
Slightly abusing notation, we denote a blank tree by $\blank$.
Intuitively, label $\blank$ signifies that whole sub-tree is~missing, and a~restriction is~the same tree after some pruning.

The $\textit{degree}$ of a node is the number of non-blank children.
Node $u$ is \emph{fully branching} if $\textit{degree}(u)=2$, \emph{dead} if $\textit{degree}(u)=0,$ and \emph{redundant} if $\textit{degree}(u)=1$.
\vspace{-9pt}
\paragraph{Parity game.} A parity game is a tuple $G = \langle V, E, \alpha, \textit{rank}, v_I\rangle$, where $V$ is the set of vertices, $v_I$ is the initial position, $E \subseteq V \times V$ is the edge relation,
$\alpha: V \to \{0,1\}$ is the partition of~the vertices between \pZ's and \pO's vertices, and ${\it rank}: V \to \mathbb{N}$ is the labelling (or colouring) of~the vertices assuming only finite number of labels.
We assume that every vertex has exactly two successors.\footnote{Notice that this assumption is not very restricting:
if a vertex has only one successor, we can simply clone that successor. This is why in the reminder of this paper, we may define (tree) games with positions
that have only one successor and assume that the definition is correct.}
Moreover, we assume that for every vertex there is specified order on successors.
Smaller of~the two successors of a vertex $v \in G.V$  will be denoted $v0$, greater $v1$.
This assumption infers function $\beta_G: V \times2^{*} \to V$, mapping finite binary sequences to the ends of finite paths in graph $\langle V,E \rangle$.
$\beta_G$ can be~inductively defined as 
\[
\beta_G(v,u) := \left\{
  \begin{array}{l l}
  v                     & \quad \text{if } u = \varepsilon,\\
  \beta_G(v i, w)         & \quad \text{if } u = i w.\\
  \end{array} \right.
\]
We extend $\beta_G$ to function $\beta^{p}_G: V \times2^{*} \to V^{*}$ relating binary sequences to paths originating from a~given vertex.
Formally, $\beta^{p}_G(\langle v, j_1 j_2 \cdots j_n \rangle) := v v_1 v_2 \cdots v_n$, where $\beta(v, j_1 j_2 \cdots j_i) = v_i$, and $i=1,2,\dots,n$.

Players $\pZ$ and $\pO$ play by moving a token, initially positioned in $v_I$, along the edges of~the graph. If the token is in a vertex $v$ such that $\alpha(v) = i$, then
\emph{Player i} chooses the next location of~the token from the set $E(v)$. A play is the path $v_0 v_1 v_2 \dots$ with $v_0 = v_I$ that was taken by~the
token as a result of~the players’ moves.
Since every vertex has at least one successor, every play is infinite.
We say that $\pZ$ wins if sequence $\textit{rank}(v_0) \textit{rank}(v_1) \textit{rank}(v_2) \cdots$ is winning.
\vspace{-9pt}
\paragraph{Tree game.} A (non-stochastic) tree game is a tuple $G = \langle V, E, \lambda, \alpha, v_I, \Phi \rangle$, where $V$ is the set of vertices,
$E \subseteq V \times V$ is the set of~edges, $\lambda: V \to \Gamma \smallsetminus \{\blank\}$ is a labelling of~the vertices and $\alpha: V \to \{0,1, \mathcal{B}\}$ is
a~partition of vertices between \pZ's, \pO's and branching vertices.
Vertex $v_I \in V$ is the initial vertex and $\Phi:\trees{\Gamma} \to [0,1] $ is a~payoff function.
As in the definition of \emph{parity games}, we assume that every vertex $v \in V$ has exactly two successors.
By $V_0$ and $V_1$ we denote the set of vertices belonging to $\pZ$ and to  $\pO$, respectively.
Set $V_{\mathcal{B}}$ is the set of branching vertices.

Players $\pZ$ and $\pO$ play by moving a token, initially positioned in $v_I$, along the edges of~the graph.
If the token is in a vertex $v$ such that $\alpha(v) = i \in \{0,1\}$, then \emph{Player i} chooses the next location of~the token from the set $E(v)$.
If the token is in a branching vertex $v$, then it splits into two indistinguishable tokens, positioned in $E(v)$, and we start two concurrent sub-games which continue their execution independently.
Therefore, result of~the players' moves, the \emph{play}, is not a path, but a tree.

An \emph{unfolding} of a game $G$ is a labelled tree $t_G: 2^{*} \to G.V$ such that $t_G(\varepsilon) = G.v_I$ and $t_G(u i) = t_G(u) i$ for all $u \in 2^{*}$, $i \in \{0,1\}$.
Notice that every game has only one unfolding.
A \emph{pre-play} $t: 2^{*} \to G.V \cup \{\blank\}$ is any restriction of~the unfolding of a game $G$ such that every node labelled with a branching vertex is fully branching and the other nodes are redundant.
The redundant nodes depict the moves of~the players.
A \emph{play} is a labelled tree $p \in \mathcal{T}_{\Gamma}$ obtained as relabelling of some pre-play $t$ so that $p(u) = \gamma[t(u) \not = \blank, G.\lambda(t(u)),\blank].$

A \emph{deterministic strategy} of $\pZ$ (resp. $\pO$) is a function $\sigma: 2^{*} \to \{0,1\}$.
The set of all deterministic strategies of $\pZ$ (resp. $\pO$) in game $G$ will be denoted $\Sigma_G$ (resp. $\Pi_G$).
We say that a tree $p$ corresponds to a strategy $\sigma \in \Sigma_G$ (resp. $\pi \in \Pi_G$) of game $G$
if $p$ is a play in which $\pZ$ (resp. $\pO$) moves accordingly to $\sigma$ (resp. $\pi$).
Notice that whenever we fix a game $G$ and strategies $\sigma \in \Sigma_G$, $\pi \in \Pi_G$, there is exactly one play that corresponds to both strategies.
We denote this tree as $G(\sigma,\pi)$.
Conversely, for every play $t$ there are strategies $\sigma \in \Sigma_G$, $\pi \in \Pi_G$ such that $t = G(\sigma,\pi)$.

It is important to notice that, in our setting, a deterministic strategy is equivalent to a function from $G.V^{+}$ into $G.V$, which is a strategy in the usual sense.
This follows directly from the existence of function $\beta^p_G: G.V \times 2^{*} \to G.V^{*}$ which relates paths in game $G$ originating from vertex $G.v_I$ to binary sequences.

Let $\Sigma_0$ be a subset of strategies of $\pZ$ in game $G$ and $\Pi_0$ be a subset of strategies of $\pO$ in game $G$.
We say that game $G$ is \emph{determined} under a profile $\langle \Sigma_0, \Pi_0 \rangle$ if the following two values are equal
\[
  \sup\limits_{\sigma \in \Sigma_0} \inf\limits_{\pi \in \Pi_0} G.\Phi(G(\sigma,\pi)) = \inf\limits_{\pi \in \Pi_0} \sup\limits_{\sigma \in \Sigma_0} G.\Phi(G(\sigma,\pi)).
\]
In that case the unique value is called the {\em value} of game $G$ under profile $\langle \Sigma_0, \Pi_0 \rangle$.
We say that game $G$ is \emph{determined under deterministic strategies} (or simply, \emph{determined}) if is determined under profile $\langle \Sigma, \Pi \rangle$.

In this paper we will only consider payoff functions defined by indicator functions of regular sets of trees and profiles consisting of deterministic strategies.
\footnote{For any further reference regarding general \emph{tree games}, one can consult \cite{mio_thesis}, chapter 4.}
In a tree game where the payoff function  is an indicator function $\chi_L$, $L$ is called the \emph{wining set}.
A \emph{tree game with regular objectives} is any tree game $G = \langle V,E,\lambda, \alpha, v_I, \chi_L\rangle$, where $L$ is a regular tree language.
We say that strategy $\sigma \in \Sigma_G$ (resp. $\pi \in \Pi_G$) is a winning strategy in game $G$, if for every strategy $\pi \in \Pi_G$ (resp. $\sigma \in \Sigma_G$) play $G(\sigma,\pi)$ belongs (resp. does not belong) to the winning set.
\vspace{-9pt}
\paragraph{Nondeterministic tree automaton (NTA)} An NTA is a tuple $\langle {\Gamma}, Q, q_I, \delta, \textit{rank} \rangle$ consisting of a finite alphabet ${\Gamma}$, a finite set of states $Q$,
a~transition function $\delta: Q \times \Gamma \to 2^{Q \times Q}$, and a rank function 
\mbox{$\textit{rank}: Q \to \mathbb{N}$}.
A \emph{run} of an NTA $A$ on tree $t$ is any labelled tree $\rho \in \trees{Q}$ such that $\rho(\varepsilon) = q_I$ and for every node $u \in 2^*$
\[
  \langle \rho(u0), \rho(u1) \rangle \in \delta(\rho(u), t(u)).
\]
A run $\rho$ is accepting, if for every infinite path $p = u_0,u_1, \dots$ where $u_i \in 2^{*}$ and $u_0=\varepsilon$ sequence $n \mapsto \textit{rank}(\rho(u_n))$ is winning.
We say that an automaton is $\mathcal{W}_{i,j}$-automaton if $i = min(\textit{rank}(Q))$
and \mbox{$j = max(\textit{rank}(Q))$}.
The pair $(i,j)$ is called the \emph{(Rabin-Mostowski) index}.
\vspace{-9pt}
\paragraph{Alternating tree automaton (ATA).}
An ATA $A$ is a tuple $A = \langle {\Gamma}, Q, q_I, \delta, \textit{rank} \rangle$, where ${\Gamma},Q,q_I$ and \emph{rank} are as previously and $\delta: Q \times {\Gamma} \to \mathcal{B}^+(\{0,1\} \times Q)$
is the transition function, where $\mathcal{B}^{+}(\{0,1\} \times Q)$ denotes the positive boolean combinations of elements from the set $\{0,1\} \times Q$.
The ATA A accepts tree $t$ if $\pZ$ has a winning strategy in the parity game $G(A,t)$ defined as:
\begin{itemize}\itemsep0em
  \item{
    $G.V := \mathcal{B}^{+}(\{0,1\} \times Q) \times 2^{*}$,
  }
  \item{
    $G.v_I := \langle \delta(q_I, t(\varepsilon)), \varepsilon \rangle.$
  }
  \item{
    If $m = \max_{q \in Q} \textit{rank}(q)$, then $\alpha$, \textit{rank}, and $E$ are defined as follows: for each $\langle \psi, w \rangle \in V$
    \begin{itemize}\itemsep0em
      \item{
        if $\psi = \psi_1 \lor \psi_2$, then $\alpha(\langle \psi, w \rangle) = 0$, $E(\langle \psi, w \rangle) = \{\langle \psi_1, w \rangle,\langle \psi_2, w \rangle \}$, and $\textit{rank}(\langle \psi, w \rangle) = m$,
      }
      \item{
        if $\psi = \psi_1 \land \psi_2$, then $\alpha(\langle \psi, w \rangle) = 1$, $E(\langle \psi, w \rangle) = \{\langle \psi_1, w \rangle,\langle \psi_2, w \rangle \}$, and $\textit{rank}(\langle \psi, w \rangle) = m$,
      }
      \item{
        if $\psi = \langle d,q \rangle$, then $\alpha(\langle \psi, w \rangle) = 0$, $E(\langle \psi, w \rangle) = \{\langle \delta(q,t(wd)), wd \rangle\}$, $\textit{rank}(\langle \psi, w \rangle)~=~\textit{rank}(q)$.
      }
    \end{itemize}
  }
\end{itemize}

ATAs are a natural syntactic extension of NTAs, and define the same class of languages.
For any further reference reader can consult, e.g., \cite{Muller1987267}, \cite{loding:habilitation}.

\emph{Game automata}, defined in \cite{DBLP:conf/fsttcs/DuparcFM11}, are ATA with the transition function $\delta$ restricted in a way that
for every letter $a$ and every state $q$, $\delta(q,a)$ has one of~the four forms: $(0,p), (1,p), (0,p) \land (1,r), (0,p) \lor (1,r)$, for some $p,r \in Q$.
\vspace{-9pt}
\paragraph{Synchronised Deterministic Tree Transducers} (SDTT) can be seen as yet another way to describe families of~regular tree languages.
Every SDTT is a tuple $\mathcal{D} = \langle {\Gamma}, Q, q_I, \alpha, \delta, \lambda \rangle$ 
consisting of a finite alphabet ${\Gamma}$, a finite set of states $Q$, an initial state $q_I$,
a transition function $\delta: Q \times {\Gamma} \to Q \times Q$, a partition of states $\alpha: Q \times {\Gamma} \to \{0,1\}$ and a relabelling $\lambda: Q \to \mathbb{N}$.
As for NTA, a \emph{run} of an SDTT $\mathcal{D}$ on tree $t$ is any labelled tree $\rho \in \trees{Q}$ such that $\rho(\varepsilon) = q_I$ and
$
  \langle \rho(u0), \rho(u1) \rangle = \delta(\rho(u), t(u)),
$ for every node $u$.
Since SDTT are deterministic, every tree $t$ admits exactly one run, denoted $\rho_{\mathcal{D}}(t).$

Slightly abusing the notation, every $\textit{SDTT}$ $\mathcal{D}$ defines a function $\mathcal{D}: \trees{{\Gamma}} \to \trees{\mathcal{D}.\lambda(Q)}$ such that $\mathcal{D}(t) = \mathcal{D}.\lambda(\rho_{\mathcal{D}}(t)).$
We say that transducer $\mathcal{D}$ accepts tree $t$ if $\pZ$ has a winning strategy in the parity game induced by $\mathcal{D}$ and $t$, i.e, in game
$G(\mathcal{D},t) = \langle 2^{*}, E, \alpha, \textit{rank}, \varepsilon \rangle$ where $E$ is the child relation, $G(\mathcal{D},t).\alpha(u) = \mathcal{D}.\alpha(\rho_{\mathcal{D}}(t)(u), t(u))$, and
$\textit{rank}(u) = \mathcal{D}(t)(u)$.
The language recognised by an SDTT $\mathcal{D}$ (denoted $L(\mathcal{D})$) is the language of all trees accepted by the transducer.

\begin{proposition}
Class of languages recognised by the SDTTs is exactly the class of languages recognised by the \emph{game automata}.
\end{proposition}
The translation is simple.
Let $\mathcal{D}$ be an SDTT, and ATA $A$ be its equivalent.
For all $q~\in~\mathcal{D}.Q$ and $a\in\mathcal{D}.\Gamma$, if $\mathcal{D}.\alpha(q,a) = i$ and $\mathcal{D}.\delta(q,a) = \langle q_0,q_1 \rangle$, then
$A.\delta(q,a) = ( 0 ,q_0 ) \diamond ( 1, q_1 )$ where $\diamond~:=~\gamma[i~=~0, \lor, \land]$.

\paragraph{Types.} Given a tree $t$ and node $u \in 2^{*}$, a context $t_{u}$ is a tree obtained from $t$ by removing sub-trees $t.u0$ and $t.u1$.
A grafting of trees $t_1, t_2$ into a context $t_u$ it the tree $t_u[t_1,t_2] = t[u0 := t_1, u1 := t_2].$
If $L$ is a~tree language, the set $t_u^{-1}L = \{ \langle t_1, t_2\rangle : t_u[t_1,t_2] \in L \}$ is called the \emph{$L$-type} of~the context $t_u$.
\vspace{-9pt}
\section{Determinacy under deterministic strategies}
\label{s:main_thm}

\begin{thm}
\label{thm:GA_determinacy}
Every tree game with regular winning set defined by a game automaton is determined under deterministic strategies.
Deciding which player has a winning strategy can be done in \emph{UP}~$\cap$~co-\emph{UP}.
\end{thm}

This theorem is an immediate consequence of Lemma \ref{lemma:GA_epsilon_determinacy} (see sect. \ref{s:main_lemma}).
Indeed, Lemma \ref{lemma:GA_epsilon_determinacy} provides an explicit polynomial reduction to \emph{parity games}.
Since \emph{parity games} are determined under positional strategies and since for a~given parity game we can decide in \emph{UP}~$\cap$~co-\emph{UP} which player has a winning strategy
(see, \cite{DBLP:journals/ipl/Jurdzinski98})
 the theorem holds.
%

As we promised in the introduction, the rest of this section explains the intuition behind the proof of~the determinacy. 
Formalisation of this intuition results in the reduction presented in Lemma \ref{lemma:GA_epsilon_determinacy}.

Tree games with regular objectives can be seen as games that are played in two phases.
First phase creates tree $t \in \trees{G.\lambda(G.V)}$, in a game-like environment.
Second phase checks whether that tree is accepted by an ATA, say $A$.
In other words, it checks whether $\pZ$ wins a game $G(A,t)$ induced by automaton $A$ and tree $t$.
Since every pre-play is a restriction of~the unfolding of $G$, we have a natural correlation between positions in game $G$ and positions in game $G(A,t).$
In fact, in both games we traverse the arenas in top-down manner and, thus, 
we could try to play those games simultaneously.
Doing so, we would obtain an infinite duration game with a parity condition with positions
of~form $\langle v,u \rangle \in G.V \times 2^{*}$ where $v$~is the node in the unfolding of $G$ and $u$ is the node in $t$.
The problem is that such game may deprive players of some information.
Indeed, in such game players would share some positions, and in those positions they would loose the information of moves of~their adversaries.
This is the reason why in tree games with general regular objectives we loose determinacy under deterministic strategies.
It happens, because there are positions in which players can choose their moves independently and concurrently.
Still, can this situation occur with objectives defined by an SDTT?

As we have stated, intuitively we are troubled in positions $\langle v , u \rangle$ that are shared by both players.
Let's assume that $\pZ$ has control over vertex $v$ and that $\pO$ controls node $u$, i.e.,  $G.\alpha(v) = 0$ and $\mathcal{D}.\alpha(\rho_{\mathcal{D}}(t), t(u)) = 1$.
Can $\pZ$ and $\pO$ choose their moves independently?

Let $q = \rho_{\mathcal{D}}(t)(u)$. The transducer is deterministic, therefore state $q$ is determined by the history, and so is the transition $\mathcal{D}.\delta(q,t(u)) = \langle q_0,q_1 \rangle$.
Since $v$ is not a branching vertex, node $u$ is redundant, i.e., one of its sub-trees is blank.
Let's assume that $\pZ$ chose $vi$ as the next move, then $t(u(1-i)) =  \blank$.
If the tree $\blank \not \in L(\mathcal{D}[q_I := q_{1-i}])$ then $\pO$ will not choose this direction in the second phase, otherwise he would forfeit the game.
Similarly, if $\blank \in L(\mathcal{D}[q_I := q_{1-i}])$, then $\pO$ will assure the victory in the second phase by choosing node $u(1-i)$.
In other words, in this situation moves in the second phase are induced by the moves in the first phase.
Therefore, players cannot choose their moves independently and we infer that they maintain perfect information throughout the game.
\vspace{-9pt}
\section{Game definable languages and $\#$-reductions}
\label{s:reductions}
Unfortunately,  not every game with regular objectives is determined under deterministic strategies. In~Theorem~\ref{prop:nonder_languages} below,
we give a simple criterion implying indeterminacy.
To present examples of indeterminate  games, it is convenient to~extend the definition of payoff function by the concept of
\mbox{$\#$-projection}.
Fortunately, such an approach leads to a slightly stronger result concerning the determinacy.

We say that a tree language $L \subseteq \trees{{\Gamma}}$ is \emph{game definable} if there is a tree game $G$, with a finite set of~vertices, such that
\[ L = L(G) := \{ G(\sigma,\pi) \in \trees{{\Gamma}}: \sigma \in \Sigma_G, \pi \in \Pi_G \}.\]


\begin{proposition}
Every \emph{game definable} tree language is recognisable by a $\mathcal{W}_{0,0}$-automaton.
\end{proposition}

If a tree language $L$ is game definable, then there exists a tree game $G$ such that $L = L(G)$.
All we need to do is to find an automaton that will accept a tree if and only if the tree is a play in game $G$.
For a~given tree $t$, the automaton will guess a pre-play whose image is $t$.
This can be done by a~$\mathcal{W}_{0,0}$-automaton because we do not need to confirm the parity condition, only the structure of~the tree.

Notice that the converse of above proposition is not true.
Indeed, there are languages recognisable by some $\mathcal{W}_{0,0}$-automata that are not {game definable}:
every language that allows different labels on the roots of trees (e.g., $\trees{{\Gamma}},$ for $|\Gamma| > 1$) is not {game definiable}.
Since $\trees{\Gamma}$ is accepted by some game automaton, this implies also that the family of languages recognisable by \emph{game automata} contains languages that are not {game definable}.
On the other hand, not every game definable language is recognisable by a game automata.
In fact, we can show that, in some sense, the structure of {game definable} languages is as rich as the structure of regular languages.
To achieve that, we use $\#$-reductions.

Let $t\in \trees{{\Gamma}}$ be a tree and $\# \not \in {\Gamma}$ be a fresh label, a \emph{$\#$-path} is any, possibly infinite, sequence $u_1, u_2,.., u_n, .. $ of nodes labelled $\#$
such that $u_{i+1}$ is a child of $u_i$ and every node in the sequence is either redundant or dead.
A tree $t_1 \in \trees{\Gamma \cup \{\#\}}$ is a \emph{partial $\#$-reduction} of a tree $t_2 \in \trees{{\Gamma} \cup \{ \# \}}$, denoted $t_1 \preceq_{\#} t_2$, if $p=u_1, u_2, \dots, u_n, \dots$ is a maximal $\#$-path in $t_2$ and
$t_1 = t_2[u_1 := t_3]$ where $t_3 = \blank$ if $p$ is infinite, $t_3 = \blank$ if $p$ is finite and the last node of $p$ is dead, or $t_3 = t_2.u$, if $p$ is finite and $u$ 
is~the non-blank child of~the last node of path $p$.
In other words, $t_1 \in \trees{{\Gamma}\cup \{\#\}}$ is a partial $\#$-reduction of $t_2 \in \trees{{\Gamma}\cup \{\#\}}$ if it is created from tree $t_2$ by collapsing some maximal $\#$-path.
On the other hand, every partial $\#$-reduction $t_1 \preceq_{\#} t_2$ defines in natural way an injection $\tau: 2^{*} \to 2^{*}$, called \emph{$\#$-injection}, that maps nodes of~the tree $t_1$ to
their original positions in tree $t_2$.

It is easy to notice that relation $\preceq_{\#}$ is strongly confluent, which means here that the shape of~the tree which is the result of collapsing
two, or more, maximal $\#$-paths does not depend on the order in which we collapse those paths.
Moreover, the reflexive-transitive closure $\preceq_{\#}^*$ of $\preceq_{\#}$ defines a partial order.
We say that $t_1 \in \trees{{\Gamma} \cup \{ \# \}}$ is an \emph{$\#$-reduction} of $t_2 \in \trees{{\Gamma} \cup \{\#\}}$ if it is the smallest tree such that $t_1\preceq_{\#}^* t_2$.
In other words, we obtain a $\#$-reduction by collapsing every $\#$-path within the original tree.
We say that $\#$-reduction is a \emph{$\#$-projection} if smaller tree has no nodes labelled with $\#$.
It is easy to check that taking a $\#$-projection of a tree $t \in \trees{{\Gamma} \cup \{\#\}}$ is a partial function $\mathcal{P}_{\#}: \trees{{\Gamma} \cup \{\#\}} \to \trees{{\Gamma} \smallsetminus \{\#\}}$
with fixed set $\trees{{\Gamma} \smallsetminus \{\#\}}$.
We extend the notion of $\#$-injection to the injections defined by $\#$-projections.
We also extend the notion of projections to languages of trees, in a non-standard way. We say that $\#$-projection of a tree language $L$ is
undefined if $\mathcal{P}_{\#}(t)$ is undefined for some tree $t \in L$, otherwise it is $\mathcal{P}_{\#}(L)$, the image of set $L$.
An example of $\#$-projections can be found in Fig. \ref{img:reductions}.

Finally, having $\#$-projections we can formulate in what sense the structure of game definable languages is similar to the structure of regular languages.
\begin{figure}
{
  \centering
  \includegraphics[width=120mm]{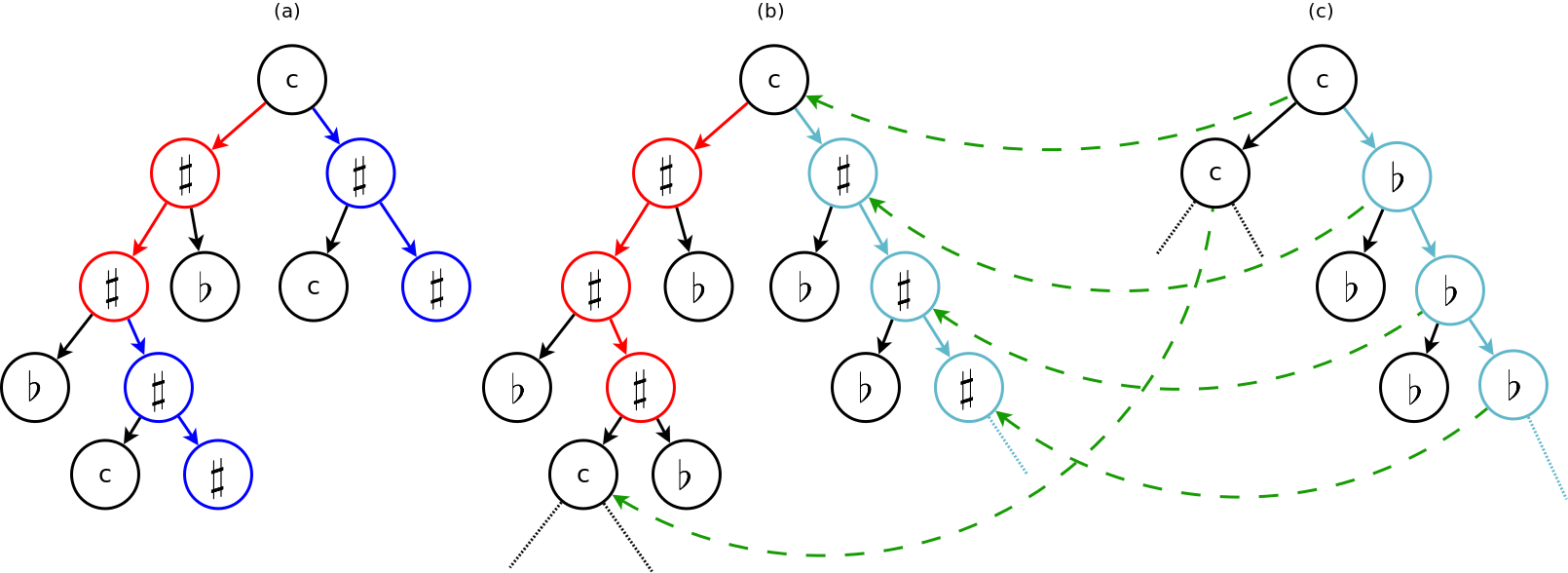}\\
  }
\caption{\footnotesize Red nodes belong to some finite $\#$-paths. Light blue nodes in (b) belong to an infinite $\#$-path. Blue nodes are labelled $\#$, but do not belong to any $\#$-path.
Part (a) depicts a tree, that has no $\#$-projection. Tree (c) is a $\#$-projection of tree (b) and the dashed arrows between (c) and (b) describe the associated $\#$-injection.}
\label{img:reductions}

  \vspace{-2ex}
\end{figure}
\begin{lemma}
\label{lemma:definable_regular}
For every non-empty regular tree language $L \subseteq \trees{{\Gamma}}$ recognisable by a $\mathcal{W}_{0,0}$-automaton there are a fresh label $\# \not \in {\Gamma}$ and a \emph{game definable} language $L' \subseteq \trees{{\Gamma} \cup 
\{ \#\}}$ such
that $L = \mathcal{P}_{\#}(L')$.
\end{lemma}

The idea behind the proof is simple -- to be accepted by a $\mathcal{W}_{0,0}$-automaton it is enough to admit a run.
Game will consists of consecutive guesses of a label and a state that are a part of a run on some tree belonging to the original language.
Auxiliary nodes, required by the guessing, will be labelled with the fresh label $\#$ that indicates redundant nodes which should be ignored by the original automaton.

\begin{proof}
Let $A$ be a non-deterministic tree automaton recognising language $L$.
Let both $A.Q^2$ and $A.{\Gamma}$ be equipped with some total order.
Let $P \subseteq  A.Q \times A.{\Gamma}$ be the set of all productive pairs, i.e., pairs $\langle q, a\rangle$ such that there
exists a tree $t$, with root labelled $a$, belonging to the language $L(A[q_I := q])$.
By $P(q) \subseteq A.{\Gamma}$ we will denote a subset of labels $P(q) = \{a \in A.{\Gamma} : \langle q,a \rangle \in P\}.$

Game $G$ is defined as follows,
$G = \langle V, E, \lambda, \alpha, \Phi \rangle$ where $V :=  (2^{A.Q \times A.Q} \times A.{\Gamma}) \cup  (A.Q \times 2^{A.{\Gamma}}),$
the set of transitions allows to guess some accepting run:
  $$
  E(\langle x,y \rangle) = \left\{
  \begin{array}{l l}
   \{ \langle q_l, P(q_l)\rangle, \langle q_r, P(q_r) \rangle  \} & \quad \text{if }  x = \{\langle q_l, q_r \rangle \}, \\
   \{ \langle x, \{min(y)\} \rangle, \langle x, y \smallsetminus \{min(y)\} \rangle  \} & \quad \text{if }  y \subseteq A.\Gamma \text{ and } |y|> 1,\\
   \{ \langle A.\delta(x,a), a\rangle \} & \quad \text{if }  y = \{a\}, \text{ where } a \in A.\Gamma, \\
   \{ \langle \{min(x)\}, y \rangle, \langle x \smallsetminus \{min(x)\}, y \rangle  \} & \quad \text{if }  x \subseteq A.Q \times A.Q \text{ and } |x|> 1.\\
  \end{array} \right.
  $$
First component begins the selection of a label associated to the state, second continues the selection up to the moment where there is only one label left.
Third component begins the selection of~the proper transition, and the last one is responsible for choosing the correct states belonging to the transition.
The partition of vertices is defined as follows: $\alpha(\langle x,y \rangle):=\gamma[x = \{\langle q_l,q_r \rangle\},\mathcal{B},0].$
The labelling of~the vertices is similar to the partition, assuring that labels from the original tree appear only on the branching vertices:
$\lambda(\langle x,y \rangle) := \gamma[x = \{\langle q_l,q_r \rangle\}, y, \#].$
Finally, the initial vertex is defined as $v_I := \langle q_I, P(q_I) \rangle.$

Notice that for every play, every node labelled with $\#$ is redundant.
Those nodes are used to guess states and labels, not to contribute to the shape of~the tree.
It is straightforward to see that the language defined by this game is similar to the original language 
in~the sense that, if~we remove redundant nodes labelled with $\#$ we will obtain a tree belonging to the original language.
That is, $L$ is $\#$-projection of $L(G)$.
The are only two problems.
The first problem is the fact that the labelling may use label $\blank$, and it is forbidden.
We solve this problem by noticing that whenever we need to use $\blank$, we can generate an infinite $\#$-path, instead.
The second problem is the fact that some vertices have no successor.
We solve that by noticing that those vertices are not reachable from the initial vertex and, therefore, may be deleted.
\end{proof}

We complete the picture of~the game definable languages by stating that the pre-image of a projection of a regular tree language is regular.

\begin{lemma}
\label{lemma:preimage}
For every label $\# \not \in {\Gamma}$ and regular language $L \subseteq \trees{{\Gamma}}$ the language $L_{\#} = \{ t \in \trees{{\Gamma} \cup \{\#\}} : \mathcal{P}_{\#}(t) \in L \}$ is regular.
Moreover, if $L$ is recognised by a $\mathcal{W}_{i,j}$-automaton, then $L_{\#}$ can be recognised by a $\mathcal{W}_{i,j}$-automaton.
\end{lemma}

Proof is~straightforward, we~modify the original automaton so~when the automaton approaches a~node labelled $\#$, it guesses that either this node belongs to an~infinite $\#$-path, and then proceeds expecting a~blank sub-tree, or
that this node belongs to~a finite $\#$-path and ignores it, ensuring that visited node is redundant.
Since such behaviour requires no parity condition, the index of~the automaton is retained.

Lemma \ref{lemma:preimage} gives us enough power to state the following.

\begin{thm}
\label{prop:nonder_languages}
For every regular language $L$, if there exist a context $t_u$ and four different trees $t_1,t_2,t_3,t_4$,
such that $t_u[t_1,t_3]\in L, t_u[t_2,t_4] \in L$, $t_u[t_2,t_3]\not\in L,$ and $t_u[t_1,t_4] \not\in L$, 
then there is a game $G$ with payoff function $(\chi_L \circ \mathcal{P}_{\varepsilon})$ that is not determined under deterministic strategies.
\end{thm}

\begin{proof}
Notice that the assumptions allow us to encode game ``matching pennies'' (cf., e.g., \cite{ApGr11}).
Indeed, we use Lemma \ref{lemma:definable_regular} to create games $G_c, G_l, G_r$ such that $\mathcal{P}_{\#}(L(G_c)) = \{t_u[\blank,\blank]\}$,
$\mathcal{P}_{\#}(L(G_l)) = \{t_1, t_2\}$ and $\mathcal{P}_{\#}(L(G_r)) = \{t_3, t_4\}$, respectively.
We can define those games so that $G_c$ has only branching vertices and both $G_l.V_1$  and $G_r.V_0$ are empty.
Finally, we create $G$ by connecting arenas of games $G_c,G_l,G_r$ so that $G.E(v) = \{ G_l.v_I, G_r.v_I\}$, where $v$ denotes the vertex corresponding to node $u$ in game $G_c$ and set $G_c.v_I$
as the initial vertex.

It is easy to show that in game $G$ each of~the players has, essentially, exactly two strategies.
$\pZ$ can choose one of two strategies $\sigma_1$ or $\sigma_2$ where $\sigma_i$ means that $\pZ$ creates tree $t_i$ in game $G_l$.
Similarly, $\pO$ has strategies $\pi_3,\pi_4$ creating trees $t_3$ and $t_4$, respectively.
Finally, we have that \mbox{$\sigma_1 \prec \pi_4 \prec \sigma_2 \prec \pi_3 \prec \sigma_1$}, where $\sigma \prec \pi$ denotes that
strategy $\pi$ wins against strategy $\sigma$.
\end{proof}

We can see above theorem as a statement saying that, if the regular winning set describes some nontrivial relation between paths in the accepted plays, then one can construct
a two player tree game that is not determined under deterministic strategies.
Of course, since single player games are trivially determined, we need to use both players to achieve the lack of determinacy.
Moreover,

\begin{proposition}
Let $G$ be a single player tree game, i.e, a tree game where $G.V_0$ or $G.V_1$ is empty, with the winning set defined by an NTA A.
If $G.V_1 = \emptyset$ (resp. $G.V_0 = \emptyset$) then the problem of deciding which player has a winning strategy belongs to $\up \cap \coup$ (resp. is $\exptime$-complete).
\end{proposition}

The proposition follows from the fact that the tree language $L(G)$ and the winning set L(A) are regular.
If $G.V_1 = \emptyset$, then $\pZ$ wins if $L(G) \cap L(A)$ is not empty.
If $G.V_0 = \emptyset$, then $\pZ$ wins if $L(G) \subseteq  L(A)$.
The non-emptiness of an NTA is in $\up \cap \coup$ and the inclusion can be decided in exponential time, cf. e.g. \cite{loding:habilitation}.
To complete the complexity results we recall that the membership problem (does $t \in L(A) ?$) requires solving a parity game and present the following lemma.

\begin{lemma}
\label{lemma:exptime}
Deciding whether $\pZ$ has a winning strategy in game $G$ with regular winning set recognised by an NTA is $\exptime$-\emph{hard}.
\end{lemma}
To prove this lemma we reduce the universality of a two letter NTA $A$.

\begin{proof}
\begin{wrapfigure}[10]{l}{25mm}
  \label{img:simple_arena}
  \vspace{-2ex}
  \centering
  \includegraphics[height=33mm]{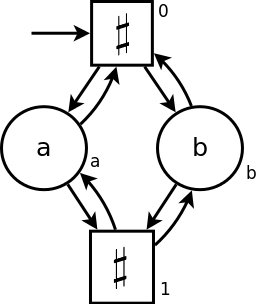}
  \caption{Arena}
  \vspace{-2ex}
\end{wrapfigure}
Let $L_{\#}$ be the pre-image of $L(A)$, as defined in Lemma \ref{lemma:preimage}, i.e., $L_\# = \{ t \in \trees{{\Gamma} \cup \{\#\}} : \mathcal{P}_{\#}(t) \in L(A) \}$, where ${\Gamma} = \{a,b\}$.
Game $G = \langle V, E, \lambda, \alpha, v_I, \Phi \rangle $ is defined as in Fig.~\ref{img:simple_arena}.
Formally, the arena has four vertices, $V = \{0,1,a,b\}$, and 8 edges: $E(0) = E(1) = \{a,b\}, E(a)= E(b) = \{0,1\}$ with $v_I = 0$.
The labelling is defined as $\lambda(a) = a, \lambda(b) = b, \lambda(0) =\lambda(1) = \#$ and the partition as $\alpha(0)=\alpha(1)=1, \alpha(a)=\alpha(b)=\mathcal{B}$.
The payoff function is the indicator function of~the regular tree language $L_{\#}$.

It is easy to check that $\pO$ has a winning strategy if and only if $L \subsetneq \trees{\{a,b\}}$ and since single player games are determined under deterministic strategies,
$\pZ$ has a winning strategy if and only if $L = \trees{\{a,b\}}$.
\end{proof}
\vspace{-9pt}
\section{Reduction to parity games}
\label{s:main_lemma}

Tree games can be seen as a certain extension of \emph{games on graphs}.
When we consider automata based objectives rather than those given by the \emph{parity condition} we loose the positional determinacy,
but, due to the fact that the \emph{parity automata} on infinite words can be determinized, we keep the determinacy under deterministic strategies.
In the case of tree games, as we have shown, for some regular winning sets we cannot guarantee that games are determined under deterministic strategies.
Yet, for game automata  -- a natural subclass of non-deterministic tree automata (cf. \cite{DBLP:conf/fsttcs/DuparcFM11}) -- we retain the determinacy.
Moreover, deciding which player has a winning strategy is not harder than in the case of parity games.

To prove that, we use the idea explained in Section \ref{s:main_thm}, which results in the following reduction of tree games with regular objectives defined by a game automata to parity games.

\begin{lemma}
\label{lemma:GA_epsilon_determinacy}
Let $G$ be a tree game.
If $G.\lambda^{-1}(\#) \subseteq V_0 \cup V_1$ and $G.\Phi = \chi_L \circ \mathcal{P}_{\#}$ for some tree language $L \subseteq \trees{{\Gamma}}$ recognisable by a \emph{game automaton},
then there exists \emph{parity game} $H$ such that $\pZ$ (resp. $\pO$) has a winning strategy in game $G$ 
if and only if $\pZ$ (resp. $\pO$) has a winning strategy in~game $H$.
The size of $H$ is polynomial with respect to the size of~the automaton recognising $L$ and to the size of~the original game.
Moreover, if $G.V$ is finite then game $H$ can be constructed in polynomial time.
\end{lemma}

We will prove this lemma in three steps.
First, we will construct game $H$ that is equivalent to $G$, is polynomial in size, but for which the cost of~the construction may by exponential.
Then we will explain how to modify game $H$ to acquire polynomial reduction.
Finally, we will prove that game $H$ is equivalent to game $G$.

\begin{proof}
Notice that for every play of game $G$, nodes labelled with $\#$ are redundant. Therefore the projection of~the language $L(G)$ is well defined.
Let $\mathcal{D}$ be an SDTT recognising language $L$.
The set $Q_{\blank} \subseteq \mathcal{D}.Q$ is the subset $\{ q \in \mathcal{D}.Q :\blank \in L(\mathcal{D}[q_I :=q]) \}$, and $m$ is the maximal rank used by $\mathcal{D}$.
Game $H$ is defined as follows:
\begin{itemize}\itemsep0em
  \item $H.V = G.V \times \mathcal{D}.Q \times \{ 0,1, ? \}$ \quad where ``?'' is an additional symbol;
  \item with $\mathcal{D}.\delta(q, G.\lambda(v)) = \langle q_0, q_1 \rangle$ and $d \in \{0,1\}$\\
  $H.E(\langle v, q, x \rangle) = \quad \quad
  \left\{ 
  \begin{array}{l l}
  \{\langle v0, q, x\rangle, \langle v1, q, x\rangle \}  & \quad \text{ if } G.\lambda(v) = \#, \\
  \{\langle v0, q_0, x\rangle, \langle v1, q_1, x\rangle \}  & \quad \text{ if } G.\lambda(v) \not = \# \text{ and } x \not = ?, \\
  \{\langle vd, q_d, \gamma[q_{1-d} \not \in Q_{\blank},?,0]\rangle\}                                                         & \quad \text{ if } G.\lambda(v) \not = \#, G.\alpha(v) = 1,\\
                                                                                                                                  & \quad \quad q \in Q_0, \text{ and } x = ?, \\  
    \{\langle vd, q_d, \gamma[q_{1-d} \in Q_{\blank},?,1]\rangle\}                                                                & \quad \text{ if } G.\lambda(v) \not = \#, G.\alpha(v) = 0,\\
                                                                                                                                  & \quad \quad q \in Q_1, \text{ and } x = ?, \\  
  \{\langle v0, q_0, x \rangle, \langle v1, q_1, x \rangle \}  & \quad \text{ otherwise;}
  \end{array} \right.
  $
  \item
  $
    H.\textit{rank}(\langle v, q, x \rangle) = \quad\left\{
  \begin{array}{l l}
  2m   & \quad \text{if } x = 0,\\
  2m+1 & \quad \text{if } x = 1,\\
  2m   & \quad \text{ if } G.\lambda(v) = \#, x = ?,\text{ and } q \in Q_{\blank}\\
  2m+1 & \quad \text{ if } G.\lambda(v) = \#, x = ?, \text{ and } q \not \in Q_{\blank}\\
  \mathcal{D}.\lambda(q)    & \quad \text{ otherwise;}
  \end{array} \right.
  $
  \item $H.\alpha(\langle v, q, x\rangle) = \left\{ 
  \begin{array}{l l}
  \mathcal{D}.\alpha(q, G.\lambda(v))  & \quad \text{ if } G.\alpha(v) = \mathcal{B}\\
  G.\alpha(v)  & \quad \text{ otherwise;}
  \end{array} \right.$
  \item $H.v_I = \langle G.v_I, \mathcal{D}.q_I, ? \rangle$
\end{itemize}

Intuitively, for $\langle v, q, x \rangle \in H.V$ first component stores the information where in the game $G$ we are,
the second component stores the information in which state of~the transducer we are, and the last component holds the information
whether the game is, sill, undecided ($x = ?$) or whether, regardless of~the future moves, \textit{Player i} won the game ($x = i,$ for $i \in \{0,1\}$).

Before we will prove the equivalence, let's consider the cost of~the reduction.
The resulting game has no more than $3 |G.V| |\mathcal{D}.Q|$ vertices and is defined by a polynomial set of equations.
The only problem is that we need to compute the set $Q_{\blank}$.
This requires solving a parity game for every state $q \in \mathcal{D}.Q$ and, unfortunately, may require exponential time.
To avoid that we modify game~$H$ 
 using a standard technique:
whenever the token is in position $\langle v, q, ? \rangle \in H.V$ whose successors are determined by the set $Q_{\blank}$ and \textit{Player i} decides to move to position
\mbox{$\langle vd,q_d, \gamma[q_{1-d} \not \in Q_{\blank},?,1-i]\rangle$},
we enter a mini-game to decide whether $q_{1-d} \in Q_{\blank}$ (see the definition of \textit{H.E}):
\textit{Player i} states whether $q_{1-d} \in Q_{\blank}$ and, after that, \textit{Player (1-i)} can either agree -- resuming game $H$ -- or can try to disagree and then
has to play game $G(\mathcal{D}[q_I = q_{1-d}],\blank)$ to prove that his/her opponent cheated.
Since the cost of implementing such mini-game is polynomial, the whole reduction is polynomial.

Now we can proceed with the proof of~the equivalence.
Since the definitions are dual, it is enough to prove the following claim.
If $\pZ$ has a winning strategy in game $H$, he/she has a winning strategy in game $G$.
Proof of~the appropriate claim for the $\pO$ is almost identical.

Let's assume that $\pZ$ has no winning strategy in game $G$, we will show that $\pZ$ has no winning strategy in game $H$.
Let $\sigma_H: 2^{*} \to \{0,1\}$ be some arbitrary strategy in game $H$ and $T_{\sigma_H}$ be the restriction of~the unfolding $T_H$ of game $H$ that is consistent with strategy $\sigma_H$.
To show that $\sigma_H$ is not a winning strategy, we need to show that there is a sequence $s_l = i_1i_2 \cdots i_n \cdots$ consistent with strategy $\sigma_H$ in game $H$ such that
sequence $H.\textit{rank}(T_{\sigma_H}(\varepsilon)) H.\textit{rank}( T_{\sigma_H}(i_1)) H.\textit{rank}( T_{\sigma_H}(i_1i_2)) \cdots H.\textit{rank}(T_{\sigma_H}(i_1i_2 \cdots i_n \cdots)) \cdots$ is loosing.
Sequence is consistent with strategy $f$ of \textit{Player i} if for every $k \in \mathbb{N}$,  $f(i_1i_2\cdots i_k) = i_{k+1}$ in the positions belonging to \textit{Player i}.
In the case of sequence $s_l$ and strategy $\sigma_H$ this translates to the statement that for every $k \in \mathbb{N}$, $T_{\sigma_H}(i_1\cdots i_k) \not = \blank$.

$\pZ$ has no winning strategy in game $G$, therefore there is  a strategy $\pi_G$ such that tree $t = \mathcal{P}_\#(G(\sigma_H, \pi_G))$ does not belong to language $L(\mathcal{D})$.
That is, for every strategy $\sigma: \Sigma_{G(\mathcal{D},t)}$ there exists sequence $s = i_1 i_2 \cdots i_n \cdots$, consistent with $\sigma$, such that sequence
$l \mapsto  \mathcal{D}(t)(i_1\cdots i_l)$ is loosing.

Let $\tau: 2^{*} \to 2^{*}$ be the injection related to $\#$-projection $t = \mathcal{P}_{\#}(T_G)$, where $T_G = G(\sigma_H, \pi_G)$.
We define $\sigma: 2^{*} \to \{0,1\}$ as a function satisfying the following conditions.
For both branching and \pZ\textit{'s} vertices of game $G$ that are mapped to vertices of $\pZ$ in game $G(\mathcal{D},t)$, $\sigma$ agrees with $\sigma_H$, i.e, 
for any $u \in 2^{*}$ such that $\mathcal{D}.\alpha(\rho_{\mathcal{D}}(t)(u), t(u)) = 0$ if $T_H (\tau(u)) = \langle q,v,x \rangle$, where $G.\alpha(v) \not = 1$, then $\sigma(u) = \sigma(\tau(u))$.

For \pO\textit{'s} vertices of game $G$ mapped to vertices of $\pZ$ in game $G(\mathcal{D},t)$ we demand that
\begin{equation}
\label{cond:sigma}
\sigma(u) = \left\{ 
  \begin{array}{l l}
  \gamma[q_0 \in Q_{\blank},0,1]  & \quad \text{ if } t(u0) = \blank,\\
  \gamma[q_1 \in Q_{\blank},1,0]  & \quad \text{ if } t(u0) \not = \blank, t(u1)  = \blank\\
  \end{array} \right.
\end{equation}
This equation will be called \emph{restriction} (\ref{cond:sigma}).
Intuitively, we demand that, whenever possible, $\sigma$ agrees with $\sigma_H$, and if the behaviour of $\sigma$ cannot be deduced from $\sigma_H$, then it behaves reasonably.

Sequence $s$ and function $\tau$ define sequence $s_2 = i_1 j_2 \cdots j_k \cdots$ such that for every $n \in \mathbb{N}$ there is a natural number $k \in \mathbb{N}$ such that
$\tau(i_1\cdots i_n) = j_1 \cdots j_k$.
Sequence $s_2$ is a witness of a loosing sequence in tree $T_{\sigma_H}$.
Before we prove it, let's notice that deterministic nature of both the unfolding and the transducer infer that for every $u \in 2^{*}$, if $T_{\sigma_H}(\tau(u)) = \langle v,q,x \rangle$,
then $\rho_{\mathcal{D}}(t)(u) = q$.

The reminder of~the proof is technical, and consists of resolving all possible instances.
Those cases depend on the vertices of game $H$ that are found on the path defined by sequence $s_2$.

Let's assume that $s_2$ is not consistent with $\sigma_H$ in game $H$, i.e., there is the smallest $k \in \mathbb{N}$ such that $T_{\sigma_H}(i_1\cdots i_{k+1}) = \blank$.
This can happen only if $T_{\sigma_H}(i_1\cdots i_k) = \langle v,q,x \rangle$, where $G.\alpha(v) = 0$ and $\mathcal{D}.\alpha(q) = 1$.
Otherwise, either $i_1\cdots i_k$ is fully branching or $s_2$ agrees with $\sigma_H$.
First, let's take care of~the $\#$ case, i.e, let's assume that $G.\lambda(v) = \#$.
With this assumption, we have that in tree $G(\sigma_H,\pi_G)$ node $i_1\cdots i_k$ belongs to an infinite $\#$-path. Otherwise the node would be deleted by the projection and $\pO$ would not be able to disagree with $\sigma_H$.
If this $\#$-path starts at the root of~the tree, then the language $L(\mathcal{D})$ does not contain a blank tree, thus $q_I \not \in Q_{\blank}$
and this infinite $\#$-path defines a loosing sequence.
If this $\#$-path does not begin at the initial vertex of game $H$ (equivalently, at the root of $T_{\sigma_H}$), 
then let $u$ be the parent of~the start of this path and $ui \sqsubseteq i_1\cdots i_k$ be it's child.
Let, $\langle w,p,y \rangle := T_{\sigma_H}(u) $ and $T_{\sigma_H}(ui) = \langle wi,q,x \rangle$.
The $\#$-path is infinite, therefore $\tau^{-1}(ui)$ is defined  and $t(\tau^{-1}(ui)) = \blank$.
Moreover, since we have a loosing sequence $s$ we have that $q \not \in Q_{\blank}$.
Depending on the value of $y$, we have three possible cases to consider.
If $y=1$, then, of course, we have a loosing sequence.
If $y = ?$, then we have a similar case as when the $\#$-path was starting at the root and the same argument is valid.
We are left with the last sub-case, $y=0$.
If $y=0$, then there is an ancestor $u_a =j_1j_2\cdots j_l \in 2^{*}$ of node $u$ labelled with $\langle w_a,p_a,? \rangle$, where $G.\alpha(w_a) = 1$ and $\mathcal{D}.\alpha(p_a) = 0$.
Moreover, one of~the sub-trees of $u_a$ is blank, let it be $i$, i.e., $j_{l+1} = (1-i)$, $T_{\sigma_H}(u_a \cdot i) = \blank$, and $T_{\sigma_H}(u_a (1-i)) = \langle w_a(1-i), p_{1-i}, 0\rangle$.
Additionally, if $\mathcal{D}.\delta(G.\lambda(w_a),p_a) = \langle p_0, p_1 \rangle$, then $p_i \in Q_{\blank}$. 
We have two cases, either $i=0$ and we broke restriction (\ref{cond:sigma}), because $p_0 \in Q_{\blank}$, or $i=1$.
If $i=1$, then either there is node $o$ such that $t(o0) \not = \blank$, $\tau(o) = u_a$ and, thus, we broke the restriction (\ref{cond:sigma}) or 
$t(o0)= \blank$ and since $s$ is loosing in $G(\mathcal{D},t)$, then $p_0 \not \in Q_{\blank}$ and we broke the restriction (\ref{cond:sigma}) again.


Now we can assume that $G.\lambda(v) \not = \#$.
Without loss of generality, we can assume that $i_{k+1} = 0$. Let $\mathcal{D}.\delta(q, G.\lambda(v)) = \langle q_0, q_1 \rangle$.
Notice that $q_0 \not \in Q_{\blank}$, otherwise the sequence $s$ would not be a loosing sequence in tree $t$ because node $j_1j_2\cdots j_n$ such that $i_1i_2\cdots i_k = \tau(j_1j_2\cdots j_n)$ would accept blank tree as a left sub-tree.
Again, we have three possibilities.
If $x = 1$, then the definition of game $H$ implies that there exists sequence with prefix $i_1i_2 \cdots i_k \cdot 1$ that is loosing.
If $x = ?$, then the non-blank successor of $i_1\cdots i_k$ is labelled $\langle v1,q_1,1 \rangle$ and, similarly, we have a loosing sequence.
Last case assumes that $x = 0$.
This implies that there is an ancestor $u \in 2^{*}$ of node $i_1\cdots i_k$ labelled with $\langle w,p,? \rangle$, where $G.\alpha(w) = 1$ and $\mathcal{D}.\alpha(p) = 0$.
This situation requires similar argument to one used in the previous paragraph.
This ends the part of~the proof, where our sequence disagrees with strategy $\sigma_H$.

Now we assume that $s_2$ is consistent with $\sigma_H$ in $H$, i.e., for all $k \in \mathbb{N}$ we have that \mbox{$T_{\sigma_H}(i_1\cdots i_{k+1}) \not = \blank$}.
Again, we have three possibilities.
If there is node $u = j_1j_2 \cdots j_k$ such that $T_{\sigma_H}(u) = \langle v,q, 1\rangle$ we obtain the loosing sequence immediately.
Otherwise, if there is a~node $u = j_1j_2 \cdots j_k$ such that $T_{\sigma_H}(u) = \langle v,q, 0\rangle$, then, as before, we find node $w$ -- the ancestor of $u$ -- and conclude that this situation cannot happen.
Last case assumes that for every $k \in \mathbb{N}$ there exists pair $\langle v,q \rangle$ such that $T_{\sigma_H}(u) = \langle v,q, ?\rangle$.
As stated before, for every $u \in 2^{*}$, whenever $T_{\sigma_H}(\tau(u)) = \langle v,q,x \rangle$ we have that $\rho_{\mathcal{D}}(t)(u) = q$, due to the deterministic nature of~the transducer.
We have to consider two cases.
If there is a node $u = j_1j_2 \cdots j_k \in 2^{*}$, for some natural $k$, such that  $u, u \cdot j_{k+1}, u \cdot j_{k+1} \cdot j_{k+2},\dots$ is an infinite $\#$-path,
then $T_{\sigma_H}(u) = \langle v,q, ?\rangle$ where $G.\lambda(v) = \#$. 
Moreover, $u$ has a corresponding node in tree $t$ that is a root of a blank sub-tree, i.e., there is a node $u_t$ such that $\tau(u_t) = u$ and $t(u_t) = \blank$.
Since $s$ is a loosing path in game $G(\mathcal{D},t)$, $\rho_{\mathcal{D}}(t)(u_t) = q \not \in Q_{\blank}$, but that implies that 
$2m+1 = H.\textit{rank}(u_j)= H.\textit{rank}(u j_{k+1}) = H.\textit{rank}(u j_{k+1} j_{k+2}) = \cdots $
and $s_2$ is a loosing sequence.
Otherwise, every $\#$-path contained in $s_2$ is finite.
Since, nodes belonging to $\#$-paths have values $2m$ or $2m+1$ they decide whether sequence is winning if and only if they belong to infinite $\#$-paths.
That means, we can exclude nodes belonging to any finite $\#$-path from the sequence 
$\varepsilon; j_1; j_1j_2; j_1j_2 \cdots j_n \cdots;  \cdots$ and, thus, we obtain an infinite sequence of nodes $u_1,u_2, \dots, u_n, \dots$
such that\\
\hspace* {3em} $H.\textit{rank}(T_{\sigma_H}(\varepsilon)) H.\textit{rank}( T_{\sigma_H}(j_1)) H.\textit{rank}( T_{\sigma_H}(j_1j_2)) \cdots H.\textit{rank}(T_{\sigma_H}(j_1j_2 \cdots j_n \cdots)) \cdots$
\\
is winning if and only if the sequence\\
\hspace* {3em}$
 H.\textit{rank}(T_{\sigma_H}(u_1)) H.\textit{rank}( T_{\sigma_H}(u_2)) H.\textit{rank}( T_{\sigma_H}(u_3)) \cdots H.\textit{rank}(T_{\sigma_H}(u_n)) \cdots
$\\
 is winning.
According to the definition of $\tau$ sequence $n \mapsto u_n$ is exactly the image of~the sequence of nodes $n \mapsto i_1i_2\cdots i_n$, where $i_1i_2\cdots i_n\cdots = s$.
As we stated before, for every $u \in 2^{*}$, whenever $T_{\sigma_H}(\tau(u)) = \langle v,q,x \rangle$ we have that $\rho_{\mathcal{D}}(t)(u) = q$, due to the deterministic nature of~the transducer, thus
sequences $n \mapsto H.\textit{rank}(T_{\sigma_H}(u_n))$ and $n \mapsto \mathcal{D}(t)(i_1i_2\cdots i_n)$ are equal.
And, since $n \mapsto \mathcal{D}(t)(i_1i_2\cdots i_n)$ is a loosing sequence $s_2$ is a loosing sequence in game $H$.\\
For every possible scenario, we have indicated a loosing sequence in three $T_{\sigma_H}$ and, thus, we are~done.

\vspace{-2ex}
\end{proof}
\vspace{-9pt}
\section{Winning strategies}
\label{s:winning_strats}
We have shown that \emph{tree games} with regular objectives are not, in general, determined under deterministic strategies.
Nevertheless, the following lemma allows us to decide whether given game, with finite arena, is determined under deterministic strategies.

\begin{lemma}
\label{lemma:winning_strategy}
For every game $G$ with a finite arena and a regular objective $\chi_L$, the set of winning strategies is regular.
Moreover, it can be recognised by an ATA of size exponential in $|A|$ and polynomial in $|G|$, where
$A$ is an ATA recognising tree language $L$.
\end{lemma}

\begin{proof}
A strategy of $\pZ$ can be seen as a restriction of~the unfolding of~the game,
restriction that is redundant in nodes belonging to $\pZ$ and fully branching in the remaining nodes.
We claim that the language of all such restrictions that represent winning strategies is~regular.
Before we proceed, notice that the language of trees that describe the set of proper strategies is regular and can be easily described by an ATA of polynomial size.
Therefore, for the clarity of the proof we will implicitly assume that every tree represents a valid strategy and we will focus on deciding whether given tree is a winning strategy.

To prove that a strategy $\sigma$ given as a tree $t_\sigma$ is not a winning strategy we need to check whether there exists a pre-play which is a redundant in nodes belonging to $\pO$ restriction of $t_\sigma$
and for which the associated play does not belong to $L$.
Intuitively, this can be done by guessing both the appropriate restriction and an accepting run of an NTA recognising language $\trees{G.V \cup \{\blank\}} \smallsetminus L$.

Now we will show how formalise above intuition, i.e, how to construct an NTA $B_{\Sigma}$ that accepts a tree iff that tree represents a strategy that is not a winning strategy.
Let $B$ be an NTA recognising language $L(B) := \trees{G.V \cup \{\blank\}} \smallsetminus L$.
We define automaton $B_{\Sigma}$ as an NTA over alphabet  $B_{\Sigma}.{\Gamma} := G.V \cup \{\blank\}$ with the set of states $B_{\Sigma}.Q := B.Q \times G.V$.
The initial state $\langle B.q_0 , G.v_0 \rangle$ implies that every run begins in the initial vertex of~the game and in the initial state of~automaton $B$.
The \emph{rank} function simply simulates that of NTA $B$, i.e., $B_{\Sigma}.\textit{rank}(\langle q,v \rangle) := B.\textit{rank}(q)$.
The transition function $B_{\Sigma}.\delta : B_{\Sigma}.Q \times B_{\Sigma}.\Gamma \to 2^{B_{\Sigma}.Q \times B_{\Sigma}.Q}$ is defined as follows.

\vspace{3mm}
\hspace{-7mm}
$B_{\Sigma}.\delta(( q,v ), a) =\vspace{-1ex}$
$$\quad
\left\{
  \begin{array}{l l}
  \{ \langle (q_0, \blank)  , (q_1, \blank)  \rangle : \langle q_0,q_1 \rangle \in B.\delta(q,\blank)\}       & \text{ if } v = \blank,\\
  \{ \langle (q_0, v0)      , (q_1, v1) \rangle : \langle q_0,q_1 \rangle \in B.\delta(q,G.\lambda(v))\}                 & \text{ if } v \in V_{\mathcal{B}} \text{ and } v = a \not = \blank,\\
  \{ \langle (q_0, v0)      , (q_1, \blank)  \rangle, \langle (q_0, \blank)  , (q_1, v1)  \rangle : \langle q_0,q_1 \rangle \in B.\delta(q, G.\lambda(v))\}                & \text{ if } v \in V_0 \cup V_1 \text{ and } v = a\not = \blank.
  \end{array} \right.
$$
Let $t_{\sigma}$ be a tree representation of some strategy $\sigma \in \Sigma$ and $\rho$ be some~accepting run.
Notice that if we project $\rho$ on its second coordinate we will obtain a pre-play, projection on the first coordinate results in a run of automaton $B$ on a play associated with that pre-play.
Indeed, the first component of above definition guarantees that the pre-play is a valid tree. The second and the third that it is, in fact, a valid pre-play.
Finally, the conditions ``$v = a \not = \blank$'' assure that this pre-play is a restriction of $t_\sigma$.

Notice that for every tree $t_\sigma$ representing a valid strategy ${\sigma} \in \Sigma$ of $\pZ$ we have that $\sigma$ is a winning strategy if and only if $t_\sigma \not \in L(B_\Sigma)$.
If $t_{\sigma} \in L(B_{\Sigma})$ then there is an accepting run $\rho$ of $B_{\Sigma}$ on $t_{\sigma}$.
Let $t$ and $\rho_t$ be trees such that for every node $u \in 2^{*}$ we have that $\rho(u) = \langle \rho_t(u), t(u) \rangle$.
Of course, $t$ is a pre-play and $\rho_t$ is an accepting run of the automaton $B$ on play associated with $t$, i.e. a play $p$ such that $p(u)= \gamma[t(u) = \blank, \blank, G.\lambda(t(u))]$ for $u \in 2^*$.
This implies that $p \in L(B)$ and, thus, $\sigma$ is not a~winning strategy.
On the other hand, if $\sigma$ is not a~winning strategy then there is~a~strategy $\pi \in \Pi$ such that the play $p = G(\sigma,\pi)$ does not belong to~the tree language $L$.
Therefore, there is an~accepting run $\rho_t$ of~$B$ on~$p$.
Finally, if $t$ is a~pre-play associated with $p$ then the tree $\rho = \langle \rho_t, t \rangle$ is an accepting run of $B_\Sigma$ on $t_\sigma$.

To end the proof of this lemma we simply complement automaton $B_{\Sigma}$.
This results in an ATA of size polynomial in $|B|$ and $|G|$.
Furthermore, this ATA can be intersected with an automaton that accepts exactly the set of trees describing valid strategies to obtain desired automaton $A_{\Sigma}$.
Since $A$ is an ATA, $|B|$ is at most exponential in $|A|$ and, thus, ATA $A_{\Sigma}$ is of exponential size.
\end{proof}
\begin{corollary}
We can decide in doubly exponential time whether a tree game $G$ with wining set defined by~an~ATA~$A$ is determined under deterministic strategies.
Moreover, deciding which player has a winning strategy is \doublyExptime-complete.
\end{corollary}

Since the non-emptiness of an ATA is $\exptime$-complete, the algorithm is an easy application of Lemma~\ref{lemma:winning_strategy}.
First, we check whether $\pZ$ has a winning strategy. If not, we complement $A$ and check whether $\pO$ has a winning strategy.
If not, then the game is not determined.

The lower bound can be obtained from known results.
Solving standard two player games with objectives defined by an LTL formulae is \doublyExptime-complete (cf. e.g. \cite{playing_box_diamonds},\cite{Pnueli:1989:SRM:75277.75293},).
Since an LTL formula can be translated into an ATA in polynomial time (e.g., see \cite{Vardi96anautomata-theoretic}) the lower bound immediately follows.

\section{Conclusions and future work}
\label{s:end}
We have shown that for the winning sets defined by \emph{game automata}, tree games are determined under deterministic strategies and that we can decide which player has a
winning strategy in $\up \cap \coup$.
Moreover, for arbitrary regular winning sets of trees, we have given a simple automata based algorithm for finding a winning strategy in doubly exponential time.

We can identify several directions for future work in this area.
One may wish to characterise regular winning objectives that guarantee determinacy under deterministic strategies:
both Theorem \ref{prop:nonder_languages} and Lemma \ref{lemma:GA_epsilon_determinacy} can be seen as the first step to such characterisation.
On the other hand, since deterministic strategies are not enough, we ask if there are larger classes of strategies that guarantee the determinacy:
if we cannot enforce that the play belongs to the winning set, can we maximise the probability?
This direction of research brings the questions about deterteminacy under randomised or mixed strategies.
\footnote{In the context of tree games, randomised strategies are not as expressive as mixed strategies (cf.\cite{mio_thesis} sect. 4.1).}
Besides the obvious question whether the games are determined, we ask about algorithms to compute, or at least approximate, the value of~the game, whenever game is determined (under some profile).

Third direction of research aims at extending our results to stochastic games.
Matteo Mio proved in his PhD thesis that \emph{stochastic meta-parity
games} are determined under deterministic strategies (for a precise statement of this result and its limitations see \cite{mio_thesis} chapter 6).
He also showed that \emph{stochastic tree games} can be expressed by \emph{non-stochastic tree games} (cf. \cite{mio_thesis} sect. 4.4), but the reduction
requires to change the payoff function in a manner that cannot be expressed by an~ATA.
Therefore, extending our results to games with stochastic positions can be an interesting direction of~research.

\emph{Acknowledgements.}
I would like to thank Christof L{\"o}ding for indicating an error in the previous version of Lemma \ref{lemma:winning_strategy} and
providing a reference for the lower bound.
Many thanks to Damian Niwiński and Teodor Knapik for valuable discussions and comments.
I also thank the anonymous referees for their helpful comments.

\vspace{-9pt}
\bibliographystyle{./eptcsstyle/eptcs}
\bibliography{vass_biblio}

\end{document}